\documentclass[conference]{IEEEtran}
\IEEEoverridecommandlockouts

\usepackage{graphicx, subfigure}
\usepackage[top=1in,bottom=1in,left=1in,right=1in]{geometry}
\usepackage[sort&compress,numbers]{natbib} 
\usepackage{amsfonts, amsmath, amssymb, amsthm, constants}
\usepackage{algorithm,algorithmic}

\usepackage{color}
\definecolor{darkred}{RGB}{100,0,0}
\definecolor{darkgreen}{RGB}{0,100,0}
\definecolor{darkblue}{RGB}{0,0,150}

\usepackage{hyperref}
\hypersetup{colorlinks=true, linkcolor=darkred, citecolor=darkgreen, urlcolor=darkblue}
\usepackage{url}

\newtheorem{thm}{Theorem}

\newtheorem{lem}{Lemma}

\def\beq{\begin{equation}}
\def\eeq{\end{equation}}
\def\beqn{\begin{eqnarray*}}
\def\eeqn{\end{eqnarray*}}
\def\bitem{\begin{itemize}}
\def\eitem{\end{itemize}}
\def\benum{\begin{enumerate}}
\def\eenum{\end{enumerate}}
\def\bmult{\begin{multline*}}
\def\emult{\end{multline*}}
\def\bcenter{\begin{center}}
\def\ecenter{\end{center}}

\newcommand{\thmref}[1]{Theorem~\ref{thm:#1}}

\newcommand{\lemref}[1]{Lemma~\ref{lem:#1}}
\newcommand{\secref}[1]{Section~\ref{sec:#1}}

\def\cN{\mathcal{N}}

\def\bA{\mathbf{A}}

\def\bI{\mathbf{I}}

\def\ba{\mathbf{a}}

\def\bu{\mathbf{u}}

\def\bx{\mathbf{x}}
\def\by{\mathbf{y}}
\def\bz{\mathbf{z}}

\def\bbE{\mathbb{E}}

\def\bbP{\mathbb{P}}
\def\bbQ{\mathbb{Q}}
\def\bbR{\mathbb{R}}

\newcommand{\E}{\operatorname{\mathbb{E}}}
\renewcommand{\P}{\operatorname{\mathbb{P}}}

\newcommand{\<}{\langle}
\renewcommand{\>}{\rangle}

\definecolor{purple}{rgb}{0.4,.1,.9}

\begin{document}

\title{Compressive Binary Search}

\author{\IEEEauthorblockN{Mark A.~Davenport}
\IEEEauthorblockA{Department of Statistics\\
Stanford University
}
\and
\IEEEauthorblockN{Ery Arias-Castro}
\IEEEauthorblockA{Department of Mathematics\\
University of California, San Diego
}}

\maketitle

\begin{abstract}
In this paper we consider the problem of locating a nonzero entry in a high-dimensional vector from possibly adaptive linear measurements.  We consider a recursive bisection method which we dub the {\em compressive binary search} and show that it improves on what any nonadaptive method can achieve.  We also establish a non-asymptotic lower bound that applies to all methods, regardless of their computational complexity. Combined, these results show that the compressive binary search is within a double logarithmic factor of the optimal performance.
\end{abstract}

\IEEEpeerreviewmaketitle

\section{Introduction}
\label{sec:intro}

How should one approach the problem of {\em finding a needle in a haystack}?  Specifically, suppose that a high-dimensional vector $\bx \in \mathbb{R}^n$ is known to have a single nonzero entry---how can we efficiently find the location of the nonzero? 
We will assume that we can learn about $\bx$ by taking $m$ noisy linear measurements of the form
\beq \label{measure}
y_i = \<\ba_i, \bx\> + z_i, \quad i = 1, \ldots, m,
\eeq
where the measurement vectors $\ba_1, \ldots, \ba_m$ have Euclidean norm at most 1 and $z_1, \ldots, z_m$ are i.i.d.\ according to $\mathcal{N}(0,1)$.  Our question reduces to the problem of choosing the vectors $\ba_1, \ldots \ba_m$ and constructing an algorithm to estimate the location of the nonzero from the measurements $y_1, \ldots, y_m$.

This is a special case of support recovery in compressive sensing (CS)~\cite{baraniuk2007compressive,intro-CS}, since \eqref{measure} is equivalent to the linear model
\beq \label{linear}
\by = \bA \bx + \bz,
\eeq
where $\by = (y_1, \dots, y_m)$, $\bA$ is the $m \times n$ matrix with row vectors $\ba_1, \dots, \ba_m$ and $\bz = (z_1, \dots, z_m)$.  (Note that the {\em rows} of $\bA$ are normalized, as opposed to the columns, which is another common convention in the CS literature.)  There are a variety of results on support recovery in the context of \eqref{linear} where the measurement matrix $\bA$ is fixed in advance (i.e., is nonadaptive) and satisfies certain desirable properties~\cite{verzelen2010minimax,iwen,5571873,5319750,CaiOMP,ZhaoYu,MR2543688}.  As an example, one can show that if $\bA$ is generated by drawing i.i.d.~$\pm 1/\sqrt{n}$ (symmetric) entries and the signal $\bx$ is 1-sparse with nonzero entry equal to $\mu > 0$, then the Lasso and Orthogonal Matching Pursuit (OMP) recover the support of $\bx$ with high probability provided that
\beq \label{nonadaptive}
\mu \ge C \sqrt{(n/m) \log n},
\eeq
with $C$ sufficiently large.  Moreover, any method based on such measurements requires $\mu$ to satisfy this lower bound for some constant $C > 0$~\cite{candes-davenport}.  This is essentially the whole story when the measurements are nonadaptive.

In contrast, suppose now that the system implementing \eqref{measure} can provide feedback in such a way as to allow for the measurements to be taken adaptively, meaning that $\ba_i$ may be chosen as a function of the observations up to time $i-1$, that is, $(y_1, \dots, y_{i-1})$.  (This implicitly assumes that $\ba_i$ is a deterministic function of this vector, but there is no loss of generality in this assumption.  See~\cite{adaptiveCS} for details.)  This instance of active or online learning has received comparatively far less attention to date.  However, in recent work~\cite{adaptiveCS} we have established lower bounds showing that {\em any} support recovery method under {\em any} adaptive sampling scheme (satisfying the conditions above) will be unable to recover the correct support unless the nonzero entry satisfies
\beq \label{adaptive}
\mu \ge C \sqrt{n/m},
\eeq
for some constant $C>0$.

Our contribution in this paper is twofold.  In \secref{binary}, we propose a {\em compressive binary search} algorithm which recursively tests whether the nonzero entry is on the left or right half of the current interval.  We show that the method reliably recovers the support of a 1-sparse vector when the nonzero entry satisfies
\beq \label{binary}
\mu \ge C \sqrt{(n/m) \log \log_2 n},
\eeq
with a constant $C>2$.  We then verify this analysis via numerical simulations.  Note that by using an adaptive measurement scheme we are able to improve upon the requirement in~\eqref{nonadaptive} by reducing the $\log n$ to $\log \log_2 n$, but our scheme does not eliminate the logarithmic factor entirely as in~\eqref{adaptive}. A corollary of this result is that in contrast to the results of~\cite{adaptiveCS}, which argued that in general adaptive strategies do not improve over nonadaptive strategies in terms of our ability to accurately recover $\bx$, we see that when $\mu$ satisfies~\eqref{binary}, adaptive strategies can significantly outperform nonadaptive ones by first identifying the location of the nonzero and then reserving a set of measurements to more accurately estimate the value of the nonzero.

In contrast to this upper bound, in \secref{lower}, we provide a simple proof that $\mu \ge C \sqrt{n/m}$ is necessary for any method to work.  This novel proof is in some sense tailored to this binary method as it too is based on testing whether the nonzero component is in the left or right half of $\bx$. In \secref{discussion}, we discuss related work in more detail and directions for future work.

\section{Compressive Binary Search}
\label{sec:binary}

\subsection{The algorithm}
The algorithm is designed assuming that the target vector $\bx$ has exactly one nonzero entry equal to $\mu > 0$; both the location and magnitude are unknown.  The methodology described here can be easily adapted to the case where the sign of the nonzero entry is unknown.  For simplicity, we assume that $n$ is dyadic, and let $s_0 = \log_2 n$, where $\log_2$ denotes the logarithm in base 2.

With a budget of $m \ge 2 \log_2 n$ measurements of the form \eqref{measure}, the binary search method proceeds as follows.  We divide our $m$ measurements into a total of $s_0$ stages, allocating
$m_s$ measurements to stage $s$, where
\beq \label{msdef}
m_s := \widetilde{m}_s+1, \qquad
\widetilde{m}_s := \left\lfloor  (m - s_0) 2^{-s}  \right\rfloor,
\eeq
where $\lfloor a \rfloor$ denotes the largest integer not greater than $a$. Note that we do not exceed our total measurement budget since
\[
\sum_{s=1}^{s_0} m_s = s_0 + \sum_{s=1}^{s_0} \widetilde{m}_s \le s_0 + (m - s_0) \sum_{s=1}^{s_0} 2^{-s} \le m.
\]
We also have $m_s \ge 1$ for all $s$, which is necessary for our algorithm to be able to run to completion.
Starting with $J_0^{(1)} := \{1, \dots, n\}$, at stage $s = 1, \dots, s_0$, we have a dyadic interval $J_0^{(s)}$ and consider its left and right halves denoted $J_1^{(s)}$ and $J_2^{(s)}$.  For example, $J_1^{(1)} := \{1, \ldots, \frac{n}{2} \}$ and $J_2^{(1)} := \{\frac{n}{2}+1, \ldots, n\}$.  Let $\bu^{(s)}$ denote the vector with entries indexed by $J_1^{(s)}$ equal to $2^{-(s_0-s+1)/2}$ and with entries indexed by $J_2^{(s)}$ equal to $-2^{-(s_0-s+1)/2}$.  Note that $\|\bu^{(s)}\| = 1$, since $|J_1^{(s)}| = |J_2^{(s)}| = 2^{s_0-s}$.  We measure $m_s$ times with $\bu^{(s)}$, meaning that we observe
\[
y_{i}^{(s)} = \< \bu^{(s)}, \bx \> + z_{i}^{(s)}, \quad i = 1, \dots, m_s.
\]
Based on these measurements, we decide between going left or right, meaning we test whether the nonzero entry is in $J_1^{(s)}$ or $J_2^{(s)}$.  We do so by simply computing
\[
w^{(s)} = \sum_{i =1}^{m_s} y_{i}^{(s)}.
\]
Specifically, we set $J_0^{(s+1)} = J_1^{(s)}$ if $w^{(s)} > 0$, and $J_0^{(s+1)} = J_2^{(s)}$ otherwise.

\subsection{Performance analysis}
\label{sec:analysis}

The binary search improves on methods based on nonadaptive measurements by by weakening the requirement \eqref{nonadaptive} to \eqref{binary}.

\begin{thm} \label{thm:binary}
In our setting, with a single nonzero entry equal to $\mu > 0$ and a measurement budget of $m \ge 2 \log_2 n$,
the probability that binary search fails to locate the nonzero entry (denoted $\bbP_e$) satisfies
\beq \label{Pe}
\bbP_e \le \frac{\log_2 n}2 \exp \left( - \frac{\mu^2 m}{8 n} \right).
\eeq
\end{thm}

\begin{proof}
Since the binary search algorithm is equivariant with respect to the ordering of the entries, we can begin by assuming without loss of generality that $\bx = (\mu, 0, \ldots, 0)^T$, i.e., the nonzero is located in the first entry of $\bx$. Thus, we can use a simple union bound to argue that
\beq \label{probbound1}
\bbP_e \le \sum_{s=1}^{s_0} \bbP\big(w^{(s)} < 0 \big),
\eeq
where $J_1^{(s)} = \{1, \ldots, 2^{-s} n\}$ and $J_2^{(s)} = \{2^{-s} n+1, \ldots 2^{1-s} n\}$.  Under our assumptions, we have that
$$
w^{(s)} \sim \cN \left( 2^{(s-1)/2} \frac{m_s \mu}{\sqrt{n}}, m_s \right).
$$
Thus we can bound
\begin{align}
\bbP\big(w^{(s)} < 0 \big) & = \bar{\Phi}\left( \mu \cdot \sqrt{ \frac{m_s 2^s}{2 n} } \right) \notag \\
 & \le \frac12 \exp \left( -\frac{ m_s \mu^2 2^s}{ n} \right), \label{probbound2}
\end{align}
since for all $t > 0$ we have
$$
\bar{\Phi}(t) := \bbP(\cN(0,1) > t) \le \frac12 \exp( - t^2/2).
$$

We next note that by construction,
$$
\widetilde{m}_s+1 \ge (m - s_0) 2^{-s}.
$$
Since $m \ge 2 s_0$, we have that $m - s_0 \ge m/2$, and hence we obtain
$$
m_s 2^s \ge (\widetilde{m}_s+1) 2^s \ge m - s_0 \ge m/2.
$$
Plugging $m_s 2^s \ge m/2$ into~\eqref{probbound2}, we obtain
$$
\bbP\big(w^{(s)} < 0 \big) \le \frac12 \exp \left( -\frac{ \mu^2 m}{8 n} \right).
$$
Plugging this into~\eqref{probbound1} we arrive at
$$
\bbP_e \le \frac{s_0}2 \exp \left(-\frac{ \mu^2 m}{8 n} \right),
$$
as desired.
\end{proof}

Note that we need \eqref{binary} with $C > 2\sqrt{2} $ for the upper bound on $\P_e$ in \eqref{Pe} to actually tend to zero as $n$ increases.  However, by taking additional measurements beyond the $2 \log_2 n$ required by this theorem, we could loosen this requirement to be able to set $C$ arbitrarily close to $2$.

\subsection{Numerical experiments}
\label{sec:numerics}

To validate our theory, we perform some simple numerical experiments. Specifically, we compare the performance of the compressive binary search procedure to that of OMP (with $\bA$ constructed with random $\pm 1/\sqrt{n}$ entries). Note that in the 1-sparse case, OMP simply reduces to identifying the column of $\bA$ most highly correlated with the measurements $\by$.  The performance of these two algorithms is shown in Figure~\ref{fig:1}, which shows the empirical probability of error as a function of $\mu$ computed by averaging over $100,000$ trials.  For these experiments, we set $n=4096$ and $m=256$.  Note that for these values of $n$ and $m$, we have that $\sqrt{n/m} = 4$ and $\sqrt{(n/m) \log \log_2 n} \approx 6.3$.  Thus, ignoring the constant terms in~\eqref{adaptive} and~\eqref{binary}, we see that the performance of the compressive binary search is largely consistent with our theory---namely, it cannot reliably identify the location of the nonzero when $\mu \le 4$ but can for $\mu \gtrsim 6.3$.  Moreover, recall that as noted in~\eqref{nonadaptive}, the nonadaptive OMP algorithm requires that $\mu$ exceed $\sqrt{(n/m) \log n} \approx 11.5$ to succeed.  Again ignoring constants, in our case this corresponds to requiring $\mu$ to be roughly $1.8$ times larger than is required for the compressive binary search procedure, and this is precisely the behavior we observe in Figure~\ref{fig:1}.

\begin{figure}[t]
   \centering
   \includegraphics[width=3.25in]{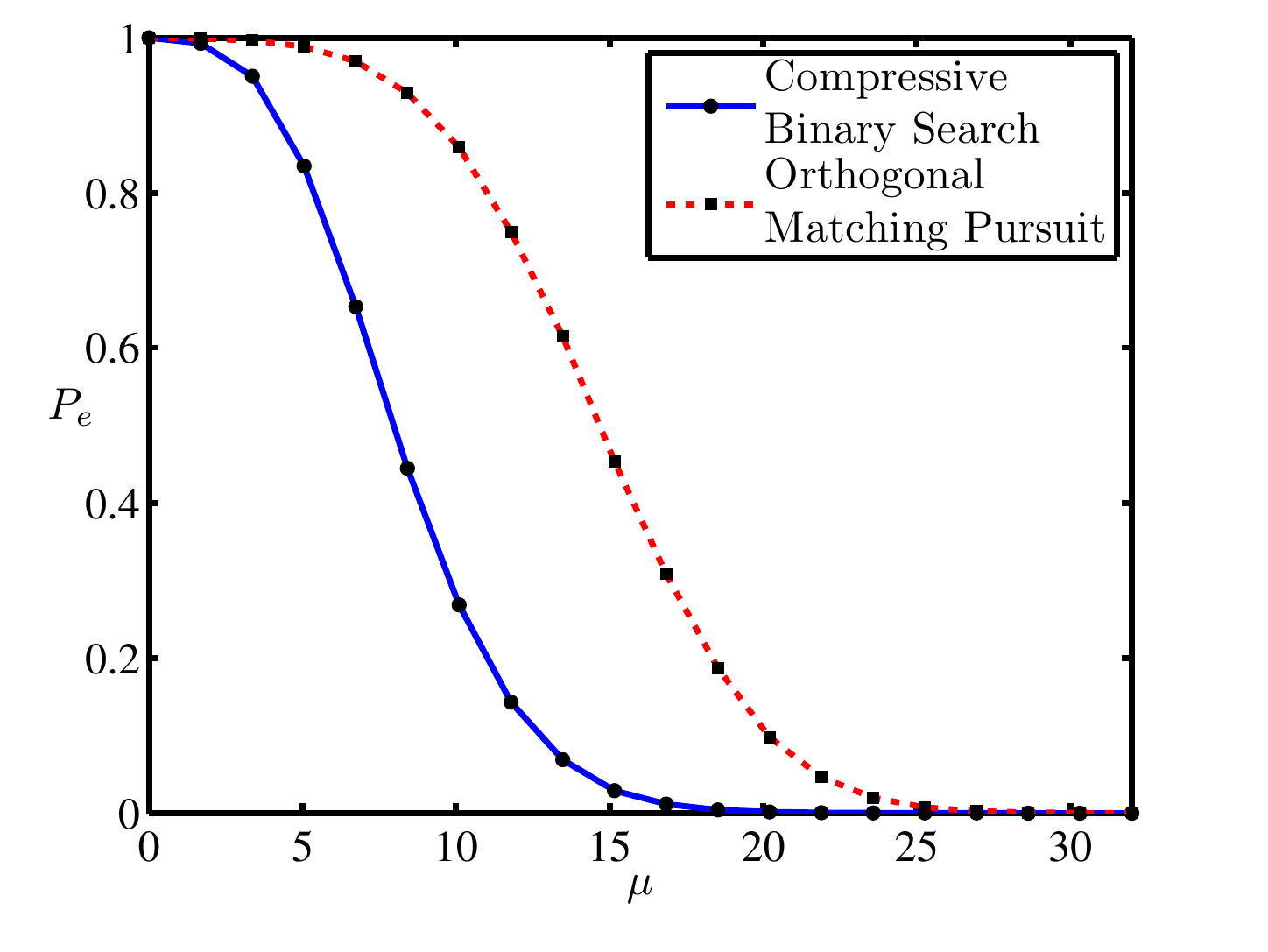}
   \caption{Comparison between compressive binary search and OMP as a function of $\mu$ for $n=4096$ and $m=256$.  Observe that compressive binary search can successfully identify the location of the nonzero for weaker values of $\mu$ than OMP, but still requires $\mu > \sqrt{n/m} = 4$.}
   \label{fig:1}
\end{figure}

\section{Lower Bound: Left or Right?}
\label{sec:lower}

We now establish an explicit, non-asymptotic lower bound for adaptive support recovery, valid for any recovery method based on adaptive measurements satisfying the conditions required here.  Though such bounds were recently derived in~\cite{adaptiveCS}, we provide a slightly simpler proof here for the case of 1-sparse signals that closely aligns with the core idea of the compressive binary search.

Let $y_{[i]} = (y_1, \dots, y_{i})$ denote the information available after taking $i$ measurements.  Let $\P_\bx$ denote the distribution of these measurements when the target vector is $\bx$.  Without loss of generality, we assume that $\ba_i$ is a deterministic function of $y_{[i-1]}$.  In that case, using the fact that $y_i$ is independent of $y_{[i-1]}$ when $\ba_i$ is given, we have
\beq \label{Px}
\P_\bx(y_{[m]}) = \prod_{i=1}^m \P_\bx(y_i  | \ba_i).
\eeq
For a subset $K \subset \{1, \dots, n\}$, let $K^c := \{1, \dots, n\} \setminus K$ and let $\bx_{K}$ be the part of $\bx$ indexed by $K$.

Let $\|\bbP - \bbQ\|_{\rm TV}$ denote the total variation metric between distributions $\P$ and $\bbQ$, and $K(\P, \bbQ)$ their Kullback-Leibler divergence~\cite{MR2319879}, related by Pinsker's inequality
\beq \label{pinsker}
\|\bbP - \bbQ\|_{\rm TV}^2 \leq \frac12 K(\bbP, \bbQ).  
\eeq

\begin{lem} \label{lem:testing}
Suppose that $n$ is even and let $J_1 = \{1, \dots, n/2\}$ and $J_2 = \{n/2+1, \dots, n\}$.  For $r=1,2$, let $\pi_r$ denote the uniform prior on the vectors $\bx \in \bbR^n$ having a single nonzero entry equal to $\mu > 0$, located in $J_r$.  Let $\P_r$ denote the distribution of $y_{[m]}$ when $\bx \sim \pi_r$.  Then
$$\|\P_2 - \P_1\|_{\rm TV}^2 \leq \frac{\mu^2 m}n.$$
\end{lem}

\begin{proof}
Let $\P_0$ denote the distribution of $y_{[m]}$ when $\bx = 0$, which is multivariate normal with zero mean and covariance $\bI$.  Using Pinsker's inequality \eqref{pinsker}, we have
\begin{align*}
\|\P_2 - \P_1\|^2_{\rm TV}
&\leq  2 \|\P_0 - \P_1\|^2_{\rm TV} + 2 \|\P_0 - \P_2\|^2_{\rm TV} \\
&\le  K(\P_0, \P_1)+K(\P_0, \P_2).
\end{align*}
Let $\bbP_{(j)}$ denote the distribution of $y_{[m]}$ when the nonzero entry (equal to $\mu$) is at $j \in \{1, \dots, n\}$.
By the law of total probability,
$$
\P_1 = \frac2n \sum_{j \in J_1} \P_{(j)},
$$
and obviously
$$
\P_0 = \frac2n \sum_{j \in J_1} \P_{0},
$$
which allows us to use the convexity of the KL divergence~\cite{MR2239987}, to obtain
$$
K(\P_0, \P_1) \leq \frac2n \sum_{j \in J_1} K(\P_{0}, \P_{(j)}).
$$
Under $\P_{(j)}$, $y_i = \mu a_{i,j} + z_i$, while under $\P_{0}$, $y_i = z_i$, so that
\begin{align*}
K(\P_{0}, \P_{(j)})
&= - \E_{0} \log \frac{\P_{(j)}}{\P_{0}} \\
&= \sum_{i=1}^m \E_{0} \left(\frac12 (y_i  - \mu a_{i,j})^2 - \frac12 y_i^2\right)  \\
&= \sum_{i=1}^m \E_{0} \left(- y_i \mu a_{i,j} + (\mu a_{i,j})^2/2\right) \\
&= \frac{\mu^2}2 \sum_{i=1}^m \E_{0} a_{i,j}^2 .
\end{align*}
The first line is by definition; the second and third are consequences of \eqref{Px} and the definition of the normal likelihood; the fourth line is because, under $\bbP_0$, $y_i$ is independent of $a_{i,j}$ and has zero mean. Hence,
$$
K(\P_0, \P_1) \leq \frac{\mu^2}{n} \sum_{i=1}^m \E_{0} \sum_{j \in J_1} a_{i,j}^2,
$$
and similarly,
\[
K(\P_0, \P_2) \leq \frac{\mu^2}{n} \sum_{i=1}^m \E_{0} \sum_{j \in J_2} a_{i,j}^2,
\]
so that
\[
K(\P_0, \P_1)+K(\P_0, \P_2) \le \frac{\mu^2}{n} \sum_{i=1}^m \E_{0} \sum_{j=1}^n a_{i,j}^2 \le \frac{\mu^2 m }{n},
\]
since $\|\ba_i\| \le 1$ for all $i$.
\end{proof}

\lemref{testing} implies a lower bound on the risk of the problem of testing whether a vector $\bx \in \bbR^n$ with a single nonzero entry equal to $\mu$ is supported on the first half or second half of the index set $\{1, \dots, n\}$.  Proving this result by directly looking at the likelihood ratio, which would be the standard approach, seems quite delicate as we are testing a mixture (supported on the first half) versus a mixture (supported on the second half).

\begin{thm} \label{thm:testing}
In the setting of \lemref{testing}, consider testing $H_1$ versus $H_2$, where $H_r$ is the hypothesis by which $\bx$ is supported in $J_r$.  Then under the uniform prior, for any test $T$,
\[
\P(T~{\rm fails}) \ge \frac12 \left(1 - \mu \sqrt{m/n} \right).
\]
\end{thm}

Note that the lower bound is also valid in the minimax sense.  In fact, the uniform prior is least favorable by invariance consideration~\cite[Sec.~8.4]{TSH}.

\begin{proof}
Under the uniform prior, we are effectively testing $\P_1$ versus $\P_2$.  The likelihood ratio test, which rejects when $L > 1$, with $L := \P_2/\P_1$, has minimum risk, bounded by
$$
\frac12 \left(1 - \|\bbP_2 - \bbP_1\|_{\rm TV}\right).
$$
(See Lemma 1 of~\cite{adaptiveCS}.) We then apply \lemref{testing} to bound the total variance distance on the RHS.
\end{proof}

\section{Discussion}
\label{sec:discussion}

Our main results can be cast as follows: \thmref{binary} implies that, with probability at least $1/4$, the binary search method locates the nonzero entry (for $n \ge 4$) if
\[
\mu \ge 4 \sqrt{\log \log_2 n} \sqrt{\frac{n}m} ,
\]
while \thmref{testing} shows that any method for locating the nonzero entry fails with probability at least $1/4$ when
\[
\mu \le \frac12 \sqrt{\frac{n}m}.
\]
Clearly, the bounds do not match.  Numerically, for $n \le 10^6$, $\log \log_2 n \le 3$, in which case the discrepancy is a multiplicative factor of $8\sqrt{3} < 14$.

We will return to the issue of whether this gap can be closed below, but first we wish to discuss an additional implication of \thmref{testing}.  Specifically, one can show that \thmref{testing} implies that for any estimator $\widehat{S}$ of the support of $\bx$, $\bbE |\widehat{S} \Delta S| \ge (1-\mu \sqrt{m/n})$. Following the same argument as in Theorem 2 of~\cite{adaptiveCS}, this implies that under the measurement model in~\eqref{linear} we have
$$
\inf_{\widehat{\bx}} \sup_{\bx : \mathrm{1-sparse} } \frac1n \bbE \| \widehat{\bx} - \bx \|^2 \ge C \frac{1}{m},
$$
where $C = 1/27$.  In contrast, \thmref{binary} implies that for sufficiently large $\mu$, there exist estimators that do far better than this bound (by a factor of $n$).

While the problems of estimating or detecting the support of a 1-sparse vector might seem to have only limited applications, in fact one can extend any algorithm that identifies the support of a 1-sparse vector to one that works for vectors with $k\ge 2$ nonzero entries.  This can be done by first exploiting a simple hashing scheme which will (with high probability) isolate each nonzero, and then applying the method for 1-sparse vectors to each hash separately.  For an overview of this approach in a similar context, see~\cite{iwen}.

We also note that~\cite{iwen} independently proposes a method very similar to the compressive binary search approach we describe.  Though~\cite{iwen} considers a different setting with continuous signals (instead of vectors as we do), the method proposed is essentially the same, except that the measurement budget is partitioned differently.  In particular, it is not obvious to us that the strategy in~\cite{iwen} will always succeed, since it does not account for rounding effects or enforce that a base number of measurements are reserved for each scale (stage) and so (to the best of our understanding) the method might exhaust its measurement budget before the algorithm terminates.  Another key difference is that by considering the simpler setting of a vector in $\mathbb{R}^n$, we can significantly simplify the analysis.  That being said, the conclusions of~\cite{iwen} are broadly similar to our own.

Finally, we also note that there a few other adaptive algorithms that have been proposed in this setting.  For example,~\cite{haupt-adaptive} proposes an algorithm similar to the compressive binary search procedure but using a different procedure for allocating measurements to each stage.  As another example, the Compressive Distilled Sensing (CDS) algorithm proposed in~\cite{haupt-compressive} considers a CS sampling algorithm which performs sequential subset selection via the random projections typical of CS.  In a different direction, \cite{4518814,4524050} suggest Bayesian approaches where the measurement vectors are sequentially chosen so as to maximize the conditional differential entropy of $y_i$ given $(y_1, \dots, y_{i-1})$.  While it remains a challenge to obtain performance bounds for the Bayesian methods suggested in \cite{4518814,4524050}, CDS is analyzed in detail in~\cite{haupt-compressive} for the task of estimating a $k$-sparse vector $\bx$.  Following the proof with a view on support recovery, one can establish that CDS is reliable in our context when
\[
\mu \ge C_n \sqrt{n/m},
\]
with $C_n \to \infty$ arbitrarily slowly, coming extremely close to the lower bound  of~\eqref{adaptive}.  However, the algorithm seems to require that $m \ge n^\alpha$ for a constant $\alpha > 0$ fixed, while binary search only requires $m \ge 2 \log_2 n$.

An important question would seem to be whether there exist methods which allow for both small $m$ {\em and} $\mu$ approaching the bound in~\eqref{adaptive}.  After the submission of this paper, Malloy and Nowak proposed a slight modification of the compressive binary search approach (involving a different allocation of the measurements to each stage) which answers this question~\cite{malloynowak}.  Specifically,~\cite{malloynowak} removes the $\log \log_2 n$ term at the cost of a slightly worse constant.  Thus, the gap between the lower bound in~\eqref{adaptive} and the upper bound in~\cite{malloynowak} differs only by a constant factor.  It would be interesting to know whether either of these bounds can be tightened.

\section*{Acknowledgements}

Thanks to E.\ Cand\`es for many insightful discussions.
M.\ D.\ is supported by NSF grant DMS-1004718.
E.\ A-C.\ is partially by ONR grant
N00014-09-1-0258.

\bibliographystyle{IEEEtran}
\bibliography{adaptiveCS}

\end{document}